\documentclass[conference, 10pt]{IEEEtran}
\usepackage{epsfig}
\usepackage{times}
\usepackage{float}
\usepackage{afterpage}
\usepackage{amsmath}
\usepackage{amstext}
\usepackage{amssymb,bm}
\usepackage{latexsym}
\usepackage{color}
\usepackage{graphicx}
\usepackage{amsmath}
\usepackage{amsthm}
\usepackage{graphicx}
\usepackage[center]{caption}
\usepackage{subfigure}
\usepackage{booktabs}
\usepackage{multicol}
\usepackage{lipsum}
\usepackage{dblfloatfix}
\usepackage{mathrsfs}
\usepackage{cite}
\allowdisplaybreaks

\newcommand{\bbC}{\mathbb{C}}

\newcommand{\bbE}{\mathbb{E}}\newcommand{\rme}{\mathrm{e}}

\newcommand{\bbI}{\mathbb{I}}

\newcommand{\bbN}{\mathbb{N}}\newcommand{\rmN}{\mathrm{N}}

\newcommand{\bbR}{\mathbb{R}}

\newcommand{\sfA}{\mathsf{A}}
\newcommand{\sfB}{\mathsf{B}}

\newcommand{\sfD}{\mathsf{D}}

\newcommand{\cD}{\mathcal{D}}

\newcommand{\cI}{\mathcal{I}}

\newcommand{\scrN}{\mathscr{N}}

\newcommand{\scrS}{\mathscr{S}}

\newcommand{\supp}{{\mathsf{supp}}}

\theoremstyle{mystyle}
\newtheorem{theorem}{Theorem}
\theoremstyle{mystyle}
\newtheorem{lemma}{Lemma}
\theoremstyle{mystyle}
\theoremstyle{mystyle}
\theoremstyle{mystyle}
\newtheorem{definition}{Definition}
\theoremstyle{remark}
\theoremstyle{mystyle}
\theoremstyle{mystyle}
\theoremstyle{mystyle}
\theoremstyle{discussion}
\theoremstyle{mystyle}
\theoremstyle{mystyle}

 \usepackage{enumitem}

\begin{document}

\title{Scalar Gaussian Wiretap Channel: Properties of the Support Size of the Secrecy-Capacity-Achieving Distribution}

\author{
\IEEEauthorblockN{ Luca Barletta$^{*}$,  Alex Dytso$^{**}$}
$^{*}$ Politecnico di Milano, Milano, 20133, Italy. Email: luca.barletta@polimi.it \\
$^{**}$ New Jersey Institute of Technology, Newark,  NJ 07102, USA.
Email: alex.dytso@njit.edu}
\maketitle
\begin{abstract}
This work studies the secrecy-capacity of a scalar-Gaussian 
wiretap channel with an amplitude constraint on the input. It is known that for this channel, the secrecy-capacity-achieving distribution is discrete with finitely many points.   This work improves such result by showing an upper bound of the order $\frac{\sfA}{\sigma_1^2}$ where $\sfA$ is the amplitude constraint and $\sigma_1^2$ is the variance of the Gaussian noise over the legitimate channel.  
\end{abstract}

\section{Introduction} 
Consider the Gaussian wiretap channel with outputs 
\begin{align}
Y_1&= X+N_1,\\
Y_2&=X+N_2,
\end{align}
where $N_1 \sim \mathcal{N}(0,\sigma_1^2)$ and  $N_2 \sim \mathcal{N}(0,\sigma_2^2)$, and with $(X,N_1,N_2)$ independent of each other.  The output $Y_1$ is observed by the legitimate receiver whereas the output $Y_2$ is observed by the malicious receiver.  In this work, we assume that the input $X$  is limited by a peak-power constraint or amplitude constraint given by $|X| \le \sfA$. 
For this setting, the secrecy-capacity is given by 
\begin{align}
C_s(\sigma_1, \sigma_2, \sfA) &= \max_{P_X:  |X| \le \sfA }  I(X; Y_1) - I(X; Y_2) \\
&= \max_{P_X:  |X| \le \sfA }   I(X; Y_1 | Y_2). \label{eq:Secracy_CAP}
\end{align}
We are interested in studying the input distribution $P_{X^\star}$ that maximizes \eqref{eq:Secracy_CAP}.  It can be shown that for $\sigma_1^2  \ge \sigma^2_2$ the secrecy-capacity is equal to zero. Therefore, in the remaining, we assume that $\sigma_1^2 < \sigma^2_2$.

\paragraph{Literature Review} 
The wiretap channel was introduced by Wyner in \cite{wyner1975wire}, who also established the secrecy-capacity  of the degraded wiretap channel. The wiretap channel plays a central role in network information theory; the interested reader is referred to \cite{bloch2011physical,Oggier2015Wiretap,Liang2009Security,poor2017wireless} and reference therein for an in-detail treatment of the topic. 

The secrecy-capacity of a Gaussian wiretap channel with an average power constraint was shown by Leung and Hellman in \cite{GaussianWireTap} where the secrecy-capacity-achieving input distribution was shown to be Gaussian.  The secrecy-capacity of Gaussian wiretap channel with an amplitude and power constraint was considered by Ozel et al. in \cite{ozel2015gaussian} where the author showed that the secrecy-capacity-achieving input distribution is discrete with finitely many points.  The work of \cite{ozel2015gaussian} was extended to noise-dependent channels by Soltani and Rezki in  \cite{soltani2018optical}.  For further studies of the properties of secrecy-capacity-achieving input distribution for a class of degraded wiretap channels, the interested reader is referred to \cite{dytso2018optimal,soltani2021degraded,nam2019secrecy}. 

The classical approach for demonstrating that the secrecy-capacity-achieving distributions are discrete relies on an analytic argument introduced to information theory by Smith in  \cite{smith1971information}. The drawback of this technique is that it does not provide any bounds on the support size of the secrecy-capacity-achieving distribution and only asserts that the support is countable.  In this work, instead of following the approach of \cite{smith1971information}, we follow the approach introduced in \cite{DytsoAmplitute2020}, which relies on the \emph{variation diminishing} property \cite{karlin1957polya}. 

This work has two goals. 
The first goal is to sharpen the results of \cite{ozel2015gaussian} by establishing a firm upper bound on the number of points in the support of the secrecy-capacity-achieving distribution.   The second goal is to study the necessary techniques required to extend the method introduced in \cite{DytsoAmplitute2020} to network information theory problems. The wiretap channel serves as an  ideal first test candidate in this research program. 

\paragraph{Outline and Contributions}
 In what follows: Section~\ref{sec:MainResults} presents our main results, which includes  two new upper bounds  on the cardinality of  the support of the optimal input distribution. 
 Section~\ref{sec:Proof} is dedicated to the proofs.  Section~\ref{sec:Conclusion}  concludes the paper with a discussion on interesting future directions.  

\paragraph{Notation} 
Throughout the paper, the deterministic scalar quantities are denoted by lower-case letters and random variables are denoted by uppercase letters.  

We denote the distribution  of a random variable $X$ by $P_{X}$. The support set of $P_X$ is denoted and defined as
\begin{align}
\supp(P_{X})&=\{x:  \text{ for every open set $ \mathcal{D} \ni x $ } \notag\\ 
&\quad \qquad \text{
 we have that $P_{X}( \mathcal{D})>0$} \}. 
\end{align} 
The relative entropy between distributions $P$ and $Q$ will be denoted by $\sfD(P\|Q)$.  The pdf of a Gaussian random variable with zero mean and variance $\sigma^2$ is denoted by  $\phi_{\sigma}( \cdot)$

Finally, the number of zeros of a function $f \colon \mathbb{R} \to \mathbb{R} $  on the interval $\cI$ is denoted by  $\rmN(\cI, f)$. Similarly, if $f  \colon \bbC \to \bbC$ is a function on the complex domain, $\rmN(\cD, f)$ denotes the number of its  zeros within the region $\cD$.

\section{Main Result} 
\label{sec:MainResults}

In this section, we state our main results.  We first present the  following ancillary lemma the first part of which was shown in \cite{ozel2015gaussian}.

\begin{lemma}\label{lem:KKT} $P_{X^\star}$ maximizes \eqref{eq:Secracy_CAP} if and only if 
\begin{align}
\Xi(x) &= C_s(\sigma_1, \sigma_2,\sfA), \, x \in \supp(P_{X^\star}),\\
\Xi(x) &\le C_s(\sigma_1, \sigma_2, \sfA), \, x \in  [-\sfA,\sfA] ,
\end{align}
where  for $x \in \mathbb{R}$
\begin{align}
\Xi(x)&=\sfD(f_{Y_1|X}(\cdot|x) \|f_{Y_1^\star})- \sfD(f_{Y_2|X}(\cdot|x) \|f_{Y_2^\star})\\
&=\bbE \left[ g(Y_1) |X=x \right] +\log\left(\frac{\sigma_2}{\sigma_1}\right), \label{eq:Writing_KKT_as_statistics}
\end{align}
and where 
\begin{align} \label{eq:functiong}
g(y)=\bbE\left[\log\frac{f_{Y_2^\star}(y+N)}{f_{Y_1^\star}(y)}\right],  \,  y\in \mathbb{R}, 
\end{align} 
with $N\sim {\cal N}(0,\sigma_2^2-\sigma_1^2)$. 
\end{lemma} 
\begin{IEEEproof}
The first part of Lemma~\ref{lem:KKT} was shown in \cite{ozel2015gaussian}. The proof of \eqref{eq:Writing_KKT_as_statistics} goes as follows: 
\begin{align}
&\sfD(f_{Y_1|X}(\cdot|x) \|f_{Y_1^\star})- \sfD(f_{Y_2|X}(\cdot|x) \|f_{Y_2^\star})-\log\left(\frac{\sigma_2}{\sigma_1}\right)\\
&=\int_{-\infty}^{\infty} \log\frac{1}{f_{Y_1^\star}(y)} \phi_{\sigma_1}(y-x) {\rm d}y \notag\\
& \quad -\int_{-\infty}^{\infty} \log\frac{1}{f_{Y_2^\star}(y)} \bbE[\phi_{\sigma_1}(y-x-N)] {\rm d}y \label{eq:intro_N} \\
&=\int_{-\infty}^{\infty} \log\frac{1}{f_{Y_1^\star}(y)} \phi_{\sigma_1}(y-x) {\rm d}y \notag\\
& \quad -\int_{-\infty}^{\infty} \bbE\left[\log\frac{1}{f_{Y_2^\star}(y+N)}\right] \phi_{\sigma_1}(y-x) {\rm d}y \label{eq:change_var} \\
&=\int_{-\infty}^{\infty} \bbE\left[\log\frac{f_{Y_2^\star}(y+N)}{f_{Y_1^\star}(y)}\right] \phi_{\sigma_1}(y-x) {\rm d}y\\
&=\int_{-\infty}^{\infty} g(y) \phi_{\sigma_1}(y-x) {\rm d}y,
\end{align}
where in~\eqref{eq:intro_N} we have introduced $N\sim {\cal N}(0,\sigma_2^2-\sigma_1^2)$; and in~\eqref{eq:change_var} we applied the change of variable $y \mapsto y+N$.  This concludes the proof. 
\end{IEEEproof}

The main result of this paper is summarized in the following theorem.

\begin{theorem}\label{thm:Main_Results}  For $\sfA>0$  
\begin{align}
| \supp(P_{X^\star})| \le \rmN\left([-R,R], g(\cdot)+\kappa_1\right) <\infty \label{eq:Implicit_Upper_BOund}
\end{align}
where 
\begin{align}
\kappa_1&=\log\left(\frac{\sigma_2}{\sigma_1}\right)-C_s,\\
 R&= \sfA \frac{\sigma_2+\sigma_1}{ \sigma_2-\sigma_1} +\sqrt{ \frac{ \frac{\sigma_2^2-\sigma_1^2}{\sigma_2^2}+2C_s}{ \frac{1}{\sigma_1^2}-\frac{1}{\sigma_2^2} } }. 
 \end{align}
 Moreover,
 \begin{align}
 \rmN\left([-R,R], g(\cdot)+\kappa_1\right)  \le   \rho \frac{\sfA^2}{\sigma_1^2} + O( \log(\sfA) ), \label{eq:Upper_Bound_Explicit}
 \end{align} 
where $\rho= (2\rme+1)^2 \left( \frac{\sigma_2+\sigma_1}{ \sigma_2-\sigma_1} \right)^2+ \left(\frac{\sigma_2+\sigma_1}{ \sigma_2-\sigma_1}+1 \right)^2$.
\end{theorem}

\section{Proofs of the Main Results} 
\label{sec:Proof}

\subsection{Proof of the bound in \eqref{eq:Implicit_Upper_BOund}} 
The function $g(\cdot)+\log\left(\frac{\sigma_2}{\sigma_1}\right)-C_s$ will play an important role in our proof in this section.  We start with the following  lemma,  which characterizes the region on which the zeros of the function $g(\cdot)+\log\left(\frac{\sigma_2}{\sigma_1}\right)-C_s$ concentrate. 
\begin{lemma}\label{Lem:boundedsupport}   Let 
\begin{equation}
	\bar{C}_s =  \frac{1}{2} \log \left(  \frac{1+ \frac{\sfA^2}{  \sigma_1^2 } }{1+\frac{\sfA^2}{\sigma_2^2}} \right).  \label{eq:Expression_Cs_power_G_Wire}
	\end{equation} 
	Then, 
	\begin{equation}
	C_s \le \bar{C}_s. 
	\end{equation} 
	Moreover, there exists some $R=R(\sigma_1,\sigma_2,\sfA)<\infty$ such that
	\begin{align}
	& \rmN\left(\mathbb{R}, g(\cdot)+\log\left(\frac{\sigma_2}{\sigma_1}\right)-C_s\right) \notag\\
	&\quad = \rmN\left([-R,R], g(\cdot)+\log\left(\frac{\sigma_2}{\sigma_1}\right)-C_s\right)<\infty.
	\end{align}
	Furthermore, $R$ can be upper-bounded as follows:
	\begin{equation}
	R \le      \sfA d_1 +d_2 \label{eq:Bound_on_R}
	\end{equation}
	where 
	\begin{align}
	d_1&=\frac{\sigma_2+\sigma_1}{ \sigma_2-\sigma_ 1},\\
	d_2&=\sqrt{ \frac{ \frac{\sigma_2^2-\sigma_1^2}{\sigma_2^2}+2 C_s}{ \frac{1}{\sigma_1^2}-\frac{1}{\sigma_2^2} } }  \le \sqrt{ \frac{ \frac{\sigma_2^2-\sigma_1^2}{\sigma_2^2}+2 \bar{C}_s}{ \frac{1}{\sigma_1^2}-\frac{1}{\sigma_2^2} } } .
	\end{align}
\end{lemma}
\begin{proof}
First,  note that
\begin{align}
C_s&= \max_{P_X:\:  |X| \le \sfA }   I(X; Y_1 | Y_2)\\
&= \max_{P_X:\:  |X| \le \sfA,\, \bbE[X^2] \le \sfA^2 }   I(X; Y_1 | Y_2)\\
&\le  \max_{P_X:\:   \bbE[X^2] \le \sfA^2 }   I(X; Y_1 | Y_2).
\end{align}
The last expression is the secrecy-capacity of a Gaussian wiretap channel with an average power constraint, which is given in \eqref{eq:Expression_Cs_power_G_Wire}.   

Second, 	for $|y|\ge \sfA$, we can lower-bound the function $g$ as follows:
	\begin{align}
	& g(y) = \bbE\left[\log f_{Y_2^\star}(y+N) \right] - \log f_{Y_1^\star}(y) \label{eq:function_g} \\
	&= \bbE\left[\log \bbE[\phi_{\sigma_2}(y+N-X^\star) | N]  \right] - \log \bbE[\phi_{\sigma_1}(y-X^\star)] \\
	&\ge \bbE\left[\log \phi_{\sigma_2}(y+N-X^\star)   \right] - \log \bbE[\phi_{\sigma_1}(y-X^\star)] \label{eq:Jensen1} \\
	&\ge \log\frac{\sigma_1}{\sigma_2}- \bbE\left[\frac{(y+N-X^\star)^2}{2\sigma_2^2}   \right] +  \frac{(|y|-\sfA)^2}{2\sigma_1^2} \label{eq:monotonicity} \\
	&=\log\frac{\sigma_1}{\sigma_2}- \bbE\left[\frac{(y-X^\star)^2}{2\sigma_2^2}   \right]-\frac{\sigma_2^2-\sigma_1^2}{2\sigma_2^2} +  \frac{(|y|-\sfA)^2}{2\sigma_1^2} \\
	&\ge \log\frac{\sigma_1}{\sigma_2}- \frac{(|y|+\sfA)^2}{2\sigma_2^2}   -\frac{\sigma_2^2-\sigma_1^2}{2\sigma_2^2} +  \frac{(|y|-\sfA)^2}{2\sigma_1^2}, \label{eq:maxdistance}
	\end{align}
	where~\eqref{eq:Jensen1} follows by applying Jensen's inequality to the first term; \eqref{eq:monotonicity} follows by
	\begin{equation}
	\bbE[\phi_{\sigma_1}(y-X^\star)] \le \phi_{\sigma_1}(|y|-\sfA), \qquad |y|\ge \sfA;
	\end{equation}
	and~\eqref{eq:maxdistance} follows by $(y-X^\star)^2 \le (|y|+\sfA)^2$ for all $|y|\ge \sfA\ge |X^\star|$. The function 
	\begin{align}
	&g(y)+\log\left(\frac{\sigma_2}{\sigma_1}\right)-C_s \notag\\
	&\ge  - \frac{(|y|+\sfA)^2}{2\sigma_2^2}   -\frac{\sigma_2^2-\sigma_1^2}{2\sigma_2^2} +  \frac{(|y|-\sfA)^2}{2\sigma_1^2}-C_s
	\end{align}
	is strictly positive when
	\begin{equation}
	|y|>  \frac{\sfA\left(\frac{1}{\sigma_1^2}+\frac{1}{\sigma_2^2}\right)+\sqrt{\frac{4\sfA^2}{\sigma_1^2 \sigma_2^2}+\left(\frac{1}{\sigma_1^2}-\frac{1}{\sigma_2^2}\right)\left(\frac{\sigma_2^2-\sigma_1^2}{\sigma_2^2}+2C_s\right)}}{\frac{1}{\sigma_1^2}-\frac{1}{\sigma_2^2}}.
	\end{equation}
By using the bound $\sqrt{a +b} \le \sqrt{a} +\sqrt{b}$, we arrive at 
\begin{align}
|y| &\ge  \sfA \frac{\sigma_2+\sigma_1}{ \sigma_2-\sigma_1} +\sqrt{ \frac{ \frac{\sigma_2^2-\sigma_1^2}{\sigma_2^2}+2C_s}{ \frac{1}{\sigma_1^2}-\frac{1}{\sigma_2^2} } }.
\end{align} 
This concludes the proof for the bound on $R$. 
\end{proof}

To show the bound on the number of points, we need to first present a number of ancillary results. 
We start with the following definition. 
\begin{definition}[Sign Changes of a Function]  The number of sign changes of a function $\xi: \Omega \to \mathbb{R}$ is given by 
\begin{equation}
  \scrS(\xi) = \sup_{m\in \bbN } \left\{\sup_{y_1< \cdots< y_m  \subseteq \Omega} \scrN \{ \xi (y_i) \}_{i=1}^m\right\} \text{,}
\end{equation}
where  $\scrN\{ \xi (y_i) \}_{i=1}^m$ is the number of changes of sign of the sequence $\{ \xi (y_i) \}_{i=1}^m $.
\end{definition}

The following theorem, shown in \cite{karlin1957polya}, will be a key step in the proof of the upper bound on the number of mass points.

\begin{theorem}[Oscillation Theorem]\label{thm:OscillationThoerem} Given domains $\bbI_1 $ and $\bbI_2$, let $p\colon \bbI_1\times \bbI_2  \to \bbR$ be a strictly totally positive
kernel.\footnote{A function $f:\bbI_1 \times \bbI_2 \to \bbR$ is said to be a totally positive kernel of order $n$ if $\det\left([f(x_i,y_j)]_{i,j = 1}^{m}\right) >0 $ for all $1\le m \le n $, and for all $x_1< \cdots < x_m \in \bbI_1  $, and $y_1< \cdots < y_m \in \bbI_2$. If $f$ is  totally positive kernel of order $n$ for all $n\in \bbN$, then $f$ is a strictly totally positive kernel.} For an arbitrary $y$, suppose $p(\cdot, y)\colon \bbI_1 \to \bbR $ is an $n$-times differentiable function. Assume that $\mu$ is a measure on $\bbI_2 $, and let $\xi \colon \bbI_2 \to \bbR $ be a function with $\scrS(\xi) = n$. For $x\in \bbI_1$, define
\begin{equation}
\Xi(x)=  \int  \xi (y) p(x ,y) {\rm d} \mu(y) \text{.} \label{eq:Integral_Transform}
\end{equation}
If $\Xi \colon \bbI_1 \to \bbR$ is an $n$-times differentiable function, then either $\rmN(\bbI_1, \Xi) \le n$, or $\Xi\equiv 0$.  
\end{theorem}

The above theorem says that the number of zeros of a function $\Xi(x)$, which is the output of integral transformation, is less than the number of sign changes of the function $  \xi (y) $, which is the input to the integral transformation.    The fact that the Gaussian pdf is a strictly totally positive kernel was show in \cite{karlin1957polya}.

We are now in the position to show the upper bound in \eqref{eq:Implicit_Upper_BOund}: 
\begin{align}
&|\supp(P_{X^\star})|  \notag\\
&\le \rmN\left([-\sfA,\sfA],   \Xi(x) - C_s(\sigma_1, \sigma_2, \sfA) \right) \label{eq:Zeros_Inclusion_Bound} \\
&= \rmN\left([-\sfA,\sfA],\bbE \left[ g(Y_1) +\log\left(\frac{\sigma_2}{\sigma_1}\right)-C_s \Big| X=x \right] \right)  \label{eq:Using_Def_of_Xi} \\
&\le \scrS\left( g(\cdot)+\log\left(\frac{\sigma_2}{\sigma_1}\right)-C_s \right) \label{eq:applying_Oscillationthm}  \\
&\le \rmN\left(\mathbb{R}, g(\cdot)+\log\left(\frac{\sigma_2}{\sigma_1}\right)-C_s\right) \\
&= \rmN\left([-R,R], g(\cdot)+\log\left(\frac{\sigma_2}{\sigma_1}\right)-C_s\right) \label{eq:Lemmaboundedsupport} \\
&< \infty, \label{eq:Follows_by_analyticity}
\end{align}
where \eqref{eq:Zeros_Inclusion_Bound} follows by using the following inclusion which is a consequence of Lemma~\ref{lem:KKT}
\begin{equation}
 \supp(P_{X^\star}) 
 \subseteq  \left\{x \in [-\sfA,\sfA] :   \Xi(x) - C_s=0 \right \} ;  \label{eq:InclusiongInequality}
\end{equation} 
\eqref{eq:Using_Def_of_Xi} follows by using \eqref{eq:Writing_KKT_as_statistics};   \eqref{eq:applying_Oscillationthm} follows by applying Theorem~\ref{thm:OscillationThoerem} and where the fact that Gaussian pdf is   a strictly totally positive kernel;  \eqref{eq:Lemmaboundedsupport} is proved in Lemma~\ref{Lem:boundedsupport};  and \eqref{eq:Follows_by_analyticity}   follows since $g(\cdot)$ is an analytic function in $(-R,R)$.

\subsection{Counting the number of zeros: Proof of the bound in \eqref{eq:Upper_Bound_Explicit} }

The key to finding an explicit upper bound will be the following  complex-analytic result. 
\begin{lemma}[Tijdeman's Number of Zeros Lemma \cite{Tijdeman1971number}]\label{lem:number of zeros of analytic function}
 Let $R, s, t$ be positive numbers such that $s>1$. For the complex valued function $f\neq  0$ which is analytic on $|z|<(st+s+t)R$, its number of zeros $  \rmN(\cD_R,f)$ within the disk $\cD_R = \{z\colon |z|\le R\} $ satisfies
\begin{align}
 & \rmN(\cD_R,f) \notag\\
  & \le \frac{1}{\log s} \left(\log \max_{|z|\le (st+s+t)R } |f(z)|   -\log \max_{|z|\le tR} |f(z)|\right) \text{.} \label{eq:Tijdeman}
\end{align}
\end{lemma}

Furthermore, the following loosened version of the bound in \eqref{eq:Implicit_Upper_BOund} will be useful. 

\begin{lemma}
\begin{align}
|\supp(P_{X^\star})| 
&\le \rmN\left([-R,R], h (\cdot) \right) +1 \label{eq:productfy1}
\end{align}
where 
\begin{align}
&\frac{h(y)}{ \sigma_1^2 f_{Y_1}(y)}  \notag\\
&=  \frac{ \bbE_N \left[  \bbE[X^\star| Y_2=y+N] \right] -y}{\sigma_2^2}-  \frac{\bbE[X^\star| Y_1=y] -y}{\sigma_1^2}  \label{eq:First_representation_h}\\
&= \frac{\bbE\left[N\log f_{Y_2}(y+N) \right]}{\sigma^2_2-\sigma^2_1} -  \frac{\bbE[X^\star| Y_1=y] -y}{\sigma_1^2},
\end{align}
and where $N\sim {\cal N}(0,\sigma_2^2-\sigma_1^2)$. 
\end{lemma}
\begin{IEEEproof}
Starting from~\eqref{eq:Lemmaboundedsupport}, we can write
\begin{align}
|\supp(P_{X^\star})| &\le \rmN\left([-R,R], g(\cdot)+\log\left(\frac{\sigma_2}{\sigma_1}\right)-C_s\right) \\
&\le \rmN\left([-R,R], g'(\cdot)\right) +1 \label{eq:Rolle}\\
&=\rmN\left([-R,R], \sigma_1^2 f_{Y_1}(\cdot) g'(\cdot)\right) +1 \label{eq:productfy1}
\end{align}
where in step \eqref{eq:Rolle} we have applied Rolle's theorem, and in step \eqref{eq:productfy1} we used the fact that multiplying by a strictly positive function (i.e., $\sigma_1^2 f_{Y_1}$) does not change the number of zeros. The first derivative of $g$ can be computed as follows:
\begin{align}
g'(y) &= \bbE\left[\frac{\rm d}{{\rm d}y}\log f_{Y_2}(y+N) \right] -\frac{ \rm d }{ {\rm d}y} \log f_{Y_1}(y) \label{eq:der1}\\
&=  \frac{ \bbE_N \left[  \bbE[X^\star| Y_2=y+N] \right] -y}{\sigma_2^2}-  \frac{\bbE[X^\star| Y_1=y] -y}{\sigma_1^2},
\end{align}
where in the last step we have used the well-known Tweedy's formula (see for example \cite{esposito1968relation,DytsoISIT2020extended}):
\begin{equation}
\bbE[X^\star| Y_i=y] = y  +\sigma^2_i\frac{ \rm d }{ {\rm d}y}\log f_{Y_i}(y).  
\end{equation}  
An alternative expression for the first term in the RHS of \eqref{eq:der1} is as follows:
\begin{align}
&\bbE\left[\frac{\rm d}{{\rm d}y}\log f_{Y_2}(y+N) \right] \notag\\
&= \int_{-\infty}^{\infty} f_N(n)\frac{ \rm d }{ {\rm d}y}\log f_{Y_2}(y+n) {\rm d}n \\
&=-\int_{-\infty}^{\infty} \left(\frac{ \rm d }{ {\rm d}n}f_N(n)\right)\cdot \log f_{Y_2}(y+n) {\rm d}n \\
&=\int_{-\infty}^{\infty} \frac{n }{\sigma^2_2-\sigma^2_1}f_N(n)\cdot \log f_{Y_2}(y+n) {\rm d}n \\
&=\frac{1}{\sigma^2_2-\sigma^2_1} \bbE\left[N\log f_{Y_2}(y+N) \right],
\end{align}
The proof is concluded by letting 
\begin{align} \label{eq:productfy2}
h(y)\triangleq \sigma_1^2 f_{Y_1}(y)g'(y).
\end{align}
\end{IEEEproof}

With the goal of getting an explicit bound on the number of zeros, through the application of Tijdeman's number of zeros Lemma, the following lemmas propose upper and lower bound to the maximum module of the complex analytic extension of $h$ over the disk ${\cal D}_R = \{z: |z|\le R \}$. 

\begin{lemma} \label{lem:moduls_upper_bound}
	Let $\breve{h}:\mathbb{C} \rightarrow\mathbb{C}$ denote the complex extension of the function $h$ in~\eqref{eq:productfy2}. Then, for $\sfB \ge \sfA$, we have that 
	\begin{align}
	&\max_{|z|\le \sfB} |\breve{h}(z)|  \le \frac{1}{\sqrt{2\pi\sigma_1^2}}  \rme^{ \frac{\sfB^2}{2\sigma_1^2} } \left( a_1 \sfB^2 +a_2 \sfB+a_3 \right)
	\end{align}
	where 
	\begin{align}
	a_1&=  \frac{ 3 \sigma_1^2}{ \sigma_2^2 \sqrt{\sigma_2^2-\sigma_1^2}}, \\
	a_2&=  \frac{ \sqrt{2} \sigma_1^2}{  \sqrt{ \sigma_2^2} \sqrt{\sigma_2^2-\sigma_1^2}} +2,\\
	a_3&=\frac{\sigma_1^2}{\sqrt{\sigma_2^2-\sigma_1^2}}  \left( \sqrt{  |\log(2\pi\sigma^2_2)|^2 +  \frac{ 24  (\sigma_2^2-\sigma_1^2)^2 }{\sigma^4_2}  + \pi^2}   \right).
	\end{align}
\end{lemma}

\begin{lemma}\label{lem:moduls_lower_bound}
	Let $\breve{h}:\mathbb{C} \rightarrow\mathbb{C}$ denote the complex extension of the function $h$ in~\eqref{eq:productfy2}. Then,
	for 
	\begin{align}
\sfB \ge  \sfA  \frac{\sigma_2^2+\sigma_1^2}{ \sigma_2^2-\sigma_1^2} \label{eq:Condition_on_B}
\end{align} 
we have that 
	\begin{align}
	\max_{|z|\le \sfB} |\breve{h}(z)| \ge  \left( c_1 \sfB   - c_2 \sfA     \right) 
 \frac{ \exp\left(  -\frac{(\sfB+\sfA)^2}{2\sigma_1^2} \hspace{-0.1cm}\right) }{\sqrt{2\pi \sigma_1^2}}	>0,
	\end{align}
	where $c_1=1-\frac{\sigma_1^2 }{\sigma_2^2} $ and $c_2=1 +  \frac{\sigma_1^2 }{\sigma_2^2}$ .
\end{lemma}

\begin{proof}

First, note that
 \begin{align}
 \frac{ \bbE_N \left[  \bbE[X^\star| Y_2=\sfB+N] \right]}{\sigma_2^2}-  \frac{\bbE[X^\star| Y_1=\sfB]}{\sigma_1^2}  
 \ge - \frac{\sfA}{\sigma_2^2}  -  \frac{\sfA}{\sigma_1^2}. \label{eq:Lower_Bound_on_CES}
 \end{align} 
Second, note that the condition in \eqref{eq:Condition_on_B} implies that 
	\begin{align}
0 \le 	\sfB\left(\frac{1}{\sigma_1^2}-\frac{1}{\sigma_2^2}\right)  - \frac{\sfA}{\sigma_2^2}  -  \frac{\sfA}{\sigma_1^2}. \label{eq:Positivity_from_condition}
\end{align}

Therefore, by using \eqref{eq:First_representation_h} together with \eqref{eq:Lower_Bound_on_CES} and \eqref{eq:Positivity_from_condition}, we arrive at
	\begin{align}
	&\max_{|z|\le \sfB} |\breve{h}(z)| \ge \left|\breve{h}(\sfB) \right|\notag\\
	 &= \hspace{-0.1cm} \left| \hspace{-0.05cm}   \frac{ \bbE \hspace{-0.05cm}\left[  \bbE[X^\star| Y_2=  \hspace{-0.05cm} \sfB+N  \hspace{-0.05cm}] \right]   \hspace{-0.05cm} -  \hspace{-0.05cm}\sfB}{\sigma_2^2}  \hspace{-0.05cm} -   \hspace{-0.05cm} \frac{\bbE[X^\star| Y_1=\sfB]-\sfB}{\sigma_1^2}  \hspace{-0.05cm} \right|  \hspace{-0.05cm}  \sigma_1^2 f_{Y_1}(\sfB)\\
	& \ge   \left(	\sfB\left(\frac{1}{\sigma_1^2}-\frac{1}{\sigma_2^2}\right)  - \frac{\sfA}{\sigma_2^2}  -  \frac{\sfA}{\sigma_1^2} \right)   \sigma_1^2 f_{Y_1}(\sfB)\\
	& \ge   \hspace{-0.05cm}  \left( \hspace{-0.05cm}	\sfB\left(\frac{1}{\sigma_1^2}-\frac{1}{\sigma_2^2} \hspace{-0.05cm}\right)  - \frac{\sfA}{\sigma_2^2}  -  \frac{\sfA}{\sigma_1^2} \right)   
\hspace{-0.1cm} \frac{\sigma_1^2 }{\sqrt{2\pi \sigma_1^2}}	 \exp\left( \hspace{-0.1cm} -\frac{(\sfB+\sfA)^2}{2\sigma_1^2} \hspace{-0.1cm}\right),
	\end{align}
where in last bound we have used  Jensen's inequality to arrive at
\begin{align}
f_{Y_1}(\sfB) &= \bbE\left[\phi_{\sigma_1}(\sfB-X^\star)\right] \\
&= \frac{1}{\sqrt{2\pi \sigma_1^2}}	\bbE\left[ \exp\left(-\frac{(\sfB-X^\star)^2}{2\sigma_1^2}\right)\right] \\
&\ge \frac{1}{\sqrt{2\pi \sigma_1^2}}	 \exp\left(-\frac{(\sfB+\sfA)^2}{2\sigma_1^2}\right). \label{eq:lb_max_h_2}
\end{align}
This concludes the proof. 
\end{proof}

With Lemma~\ref{lem:moduls_upper_bound} and Lemma~\ref{lem:moduls_lower_bound} at our disposal we are now ready to used Tijdeman's Number of Zeros Lemma  to provide an upper bound on the number of mass points:
\begin{align}
&\rmN\left([-  R,  R],h(\cdot)\right)\\
&\le \rmN\left( \mathcal{D}_{R},\breve{h}(\cdot) \right) \label{eq:Extension_to_complex_bound}\\
 &\le  \min_{s> 1,\, t > 0 } \left\{ \frac{\log  \frac{\max_{|z| \le (st+s+t)R}|\breve h (z)| }{\max_{|z| \le t R} |\breve h (z)|}}{\log s}  \right\} \label{frm:lem: bd on the no of zeros of analytic func} \\
      &\le  \log  \frac{\frac{  \rme^{ \frac{(2\rme+1)^2R^2}{2\sigma_1^2} } }{\sqrt{2\pi\sigma_1^2}} \left( a_1 (2\rme+1)^2R^2 +a_2(2\rme+1)R+a_3 \right)
 }{ \left( c_1 R   - c_2 \sfA     \right) 
 \frac{  \exp\left(  -\frac{(R+\sfA)^2}{2\sigma_1^2} \hspace{-0.1cm}\right) }{\sqrt{2\pi \sigma_1^2}} } \label{eq:choosing_s_and_t}\\
  &= \frac{(2\rme+1)^2R^2}{2\sigma_1^2}+ \frac{(R+\sfA)^2}{2\sigma_1^2}   \notag\\
  & \quad +   \log  \frac{  a_1 (2\rme+1)^2R^2 +a_2(2\rme+1)R+a_3 
 }{  c_1 R   - c_2 \sfA      } \\
 &= \frac{(2\rme+1)^2(d_1 \sfA +d_2)^2}{2\sigma_1^2}+ \frac{( (d_1+1)\sfA+d_2)^2}{2\sigma_1^2}   \notag\\
  &  +   \log  \frac{  a_1 (2\rme+1)^2(d_1 \sfA +d_2)^2 +a_2(2\rme+1)(d_1 \sfA +d_2)+a_3 
 }{ (c_1d_1-c_2) \sfA +c_1 d_2   }  \label{eq:inserting_R_def}\\
 & \le   b_1  \frac{\sfA^2}{\sigma^2_1}+ b_2 +\log \frac{b_3\sfA^2+b_4 \sfA+b_5}{ b_6 \sfA+b_7}, \label{eq:square_bound}\\
  & \le   b_1  \frac{\sfA^2}{\sigma^2_1}+ O( \log(\sfA) ) ,  \label{eq:Big_O_bound}
      \end{align}
\eqref{eq:Extension_to_complex_bound} follows since extending to large domain can only increase the number of zeros; \eqref{frm:lem: bd on the no of zeros of analytic func} follows by the Tijdeman's Number of Zeros Lemma;      \eqref{eq:choosing_s_and_t} follows by choosing $s=\rme$ and $t=1$ and using bounds in Lemma~\ref{lem:moduls_upper_bound} and Lemma~\ref{lem:moduls_lower_bound}; \eqref{eq:inserting_R_def} follows using the value of $R$ in \eqref{eq:Bound_on_R}; \eqref{eq:square_bound} using the bound $(a+b)^2 \le 2 (a^2+b^2)$ and defining
\begin{align}
b_1&= (2\rme+1)^2d_1^2+(d_1+1)^2\\
&=  (2\rme+1)^2 \left( \frac{\sigma_2+\sigma_1}{ \sigma_2-\sigma_1} \right)^2+ \left(\frac{\sigma_2+\sigma_1}{ \sigma_2-\sigma_1}+1 \right)^2  \\
b_2&= \frac{((2\rme+1)^2+1)d_2^2}{\sigma_1^2} \\
&= \frac{((2\rme+1)^2+1)}{\sigma_1^2}   \frac{ \frac{\sigma_2^2-\sigma_1^2}{\sigma_2^2}+2C_s}{ \frac{1}{\sigma_1^2}-\frac{1}{\sigma_2^2} } \\
&=((2\rme+1)^2+1) \left (1+ 2\frac{\sigma_2^2}{\sigma_2^2-\sigma_1^2} C_s \right)\\
b_3&=2 (2 \rme+1)^2 a_1 d_1^2\\
&=    2 (2 \rme+1)^2 \frac{ 3 \sigma_1^2}{ \sigma_2^2 \sqrt{\sigma_2^2-\sigma_1^2}} \left(\frac{\sigma_2+\sigma_1}{ \sigma_2-\sigma_ 1} \right)^2\\
b_4&=(2 \rme+1) d_1  a_2\\
&=  (2 \rme+1)  \frac{\sigma_2+\sigma_1}{ \sigma_2-\sigma_ 1} \left( \frac{ \sqrt{2} \sigma_1^2}{  \sqrt{ \sigma_2^2} \sqrt{\sigma_2^2-\sigma_1^2}} +2 \right) \\
b_5&=2(2\rme+1)^2 a_1d_2^2+  (2\rme +1) a_2 d_2 +a_3\\
&=2(2\rme+1)^2  \frac{ 3 \sigma_1^2}{ \sigma_2^2 \sqrt{\sigma_2^2-\sigma_1^2}} \left(  \frac{ \frac{\sigma_2^2-\sigma_1^2}{\sigma_2^2}+2C_s}{ \frac{1}{\sigma_1^2}-\frac{1}{\sigma_2^2} } \right) \notag\\
&+  (2\rme +1) \left(\frac{ \sqrt{2} \sigma_1^2}{  \sqrt{ \sigma_2^2} \sqrt{\sigma_2^2-\sigma_1^2}} +2 \right) \sqrt{ \frac{ \frac{\sigma_2^2-\sigma_1^2}{\sigma_2^2}+2C_s}{ \frac{1}{\sigma_1^2}-\frac{1}{\sigma_2^2} } }   \notag\\
& +\frac{\sigma_1^2}{\sqrt{\sigma_2^2-\sigma_1^2}} \cdot  \sqrt{  |\log(2\pi\sigma^2_2)|^2 +  \frac{ 24  (\sigma_2^2-\sigma_1^2)^2 }{\sigma^4_2}  + \pi^2}   \\
b_6&= c_1d_1-c_2\\
&= \frac{\sigma_2^2-\sigma_1^2 }{\sigma_2^2} \frac{\sigma_2+\sigma_1}{ \sigma_2-\sigma_ 1}  -  \frac{\sigma_2^2+\sigma_1^2 }{\sigma_2^2} =2 \frac{\sigma_1}{\sigma_2}\\
b_7&=c_1 d_2\\
&=  \frac{\sigma_2^2-\sigma_1^2 }{\sigma_2^2}\sqrt{ \frac{ \frac{\sigma_2^2-\sigma_1^2}{\sigma_2^2}+2C_s}{ \frac{1}{\sigma_1^2}-\frac{1}{\sigma_2^2} } };
\end{align}
and  \eqref{eq:Big_O_bound} follows from the fact that the  $b_1,b_3,b_4$ and $b_6$ coefficients  do not depend $\sfA$ and the fact that the coefficients  $b_2,b_5$ and $b_4$, while do  depend on $\sfA$ through $C_s$, do not grow with $\sfA$.  The fact that $C_s$ does not grow with $\sfA$ follows from the bound in \eqref{eq:Expression_Cs_power_G_Wire}.

\section{Conclusion} 
\label{sec:Conclusion}

This works has focused on deriving upper bounds on the number of mass points of secrecy-capacity-achieving distribution. 

The upper bounds in  Theorem~\ref{thm:Main_Results} are generalizations of  the upper bounds on the number of points presented in \cite{DytsoAmplitute2020} in the context of a point-to-point additive white Gaussian noise (AWGN) channel with an amplitude constraint.  Indeed, if we let $\sigma_2 \to \infty$, while keeping $\sigma_1$ and $\sfA$ fixed, then the wiretap channel reduces to the AWGN point-to-point channel.

An interesting future direction would be to  find a matching implicit lower bound in \eqref{eq:Implicit_Upper_BOund}.  In \cite{DytsoAmplitute2020} such a matching lower bound was found and shown  to be tight with a multiplicative factor of two from the upper bound.  
 These results effectively show that the oscillation theorem (see Theorem~\ref{thm:OscillationThoerem}) is a strong enough tool for producing upper bounds on the  cardinality of secrecy-capacity-achieving distributions for point-to-point channels. A matching lower bound in the case of the wiretap channel would demonstrate that oscillation theorem can also play an important role in network information theory problems.  In \cite{DytsoAmplitute2020}, the key tool to finding the lower bound was the observation that a linear combination of $n+1$ distinct Gaussians with distinct variances can have at most  $2n$ zeros. In the wiretap channel, due to a more complicated structure of the function $g$ in~\eqref{eq:functiong}, it is not immediately clear how such an argument can be applied.  

It will also be interesting to augment an explicit upper bound on the number of points in \eqref{eq:Upper_Bound_Explicit} with a lower bound on the number of points.  A possible line of attack consists of the following steps: 
\begin{align}
C_s(\sigma_1, \sigma_2, \sfA) &= I(X^\star;Y_1)- I(X^\star; Y_2)\\
& \le H(X^\star)- I(X^\star; Y_2)\\
& \le \log( | \supp(P_{X^\star})| ) - I(X^\star; Y_2), \label{eq:Step_1}
\end{align} 
where the above uses the non-negativity of entropy and the fact that entropy is maximized by a  uniform distribution.  Furthermore, by using a suboptimal uniform (continuous) distribution on $[-\sfA,\sfA]$ as an input and the entropy power inequality, the secrecy-capacity can be lower-bounded by 
\begin{equation}
C_s(\sigma_1, \sigma_2, \sfA)  \ge   \frac{1}{2} \log \left( 1+ \frac{ \frac{2 \sfA^2}{ \pi \rme \sigma_1^2 } }{1+\frac{\sfA^2}{\sigma_2^2}} \right).  \label{eq:Step_2}
\end{equation} 
Combing  bounds in \eqref{eq:Step_1} and \eqref{eq:Step_2} we arrive at the following lower bound on the number of points:
\begin{align}
 | \supp(P_{X^\star})| \ge    \sqrt{1+ \frac{ \frac{2 \sfA^2}{ \pi \rme \sigma_1^2 } }{1+\frac{\sfA^2}{\sigma_2^2}}} \rme^{  I(X^\star; Y_2) } . 
 \end{align}
 At this point one needs to determine the behavior of $I(X^\star; Y_2)$.  A trivial lower bound on  $ | \supp(P_{X^\star})| $ can be found by lower bounding $ I(X^\star; Y_2)$ by zero. However, this lower bound  on $ | \supp(P_{X^\star})| $  does not grow with $\sfA$ while the upper bound increases with $\sfA$. 
A possible way of establishing a lower bound that is increasing in $\sfA$ is by showing that $ I(X^\star; Y_2) \approx  \frac{1}{2} \log \left(1+\frac{\sfA^2}{\sigma_2^2} \right) $.  However, because not much is known about the structure of the optimal input distribution $P_{X^\star}$, it is not immediately evident how one can establish such an approximation or whether it is valid.

\bibliographystyle{IEEEtran}
\bibliography{refs.bib}

\begin{thebibliography}{10}
\providecommand{\url}[1]{#1}
\csname url@samestyle\endcsname
\providecommand{\newblock}{\relax}
\providecommand{\bibinfo}[2]{#2}
\providecommand{\BIBentrySTDinterwordspacing}{\spaceskip=0pt\relax}
\providecommand{\BIBentryALTinterwordstretchfactor}{4}
\providecommand{\BIBentryALTinterwordspacing}{\spaceskip=\fontdimen2\font plus
\BIBentryALTinterwordstretchfactor\fontdimen3\font minus
  \fontdimen4\font\relax}
\providecommand{\BIBforeignlanguage}[2]{{%
\expandafter\ifx\csname l@#1\endcsname\relax
\typeout{** WARNING: IEEEtran.bst: No hyphenation pattern has been}%
\typeout{** loaded for the language `#1'. Using the pattern for}%
\typeout{** the default language instead.}%
\else
\language=\csname l@#1\endcsname
\fi
#2}}
\providecommand{\BIBdecl}{\relax}
\BIBdecl

\bibitem{wyner1975wire}
A.~D. Wyner, ``The wire-tap channel,'' \emph{Bell Syst. Tech. J.}, vol.~54,
  no.~8, pp. 1355--1387, 1975.

\bibitem{bloch2011physical}
M.~Bloch and J.~Barros, \emph{Physical-Layer Security:From Information Theory
  to Security Engineering}.\hskip 1em plus 0.5em minus 0.4em\relax Cambridge
  University Press, 2011.

\bibitem{Oggier2015Wiretap}
F.~Oggier and B.~Hassibi, ``A perspective on the {MIMO} wiretap channel,''
  \emph{Proc. of {IEEE}}, vol. 103, no.~10, pp. 1874--1882, 2015.

\bibitem{Liang2009Security}
Y.~Liang, H.~V. Poor, and S.~{Shamai (Shitz)}, ``Information theoretic
  security,'' \emph{Foundations and Trends in Communications and Information
  Theory}, vol.~5, no. 4--5, pp. 355--580, 2009.

\bibitem{poor2017wireless}
H.~V. Poor and R.~F. Schaefer, ``Wireless physical layer security,''
  \emph{Proc. the Natl. Acad. Sci. U.S.A.}, vol. 114, no.~1, pp. 19--26, 2017.

\bibitem{GaussianWireTap}
S.~Leung-Yan-Cheong and M.~Hellman, ``The {G}aussian wire-tap channel,''
  \emph{IEEE Trans. Inf. Theory}, vol.~24, no.~4, pp. 451--456, 1978.

\bibitem{ozel2015gaussian}
O.~Ozel, E.~Ekrem, and S.~Ulukus, ``Gaussian wiretap channel with amplitude and
  variance constraints,'' \emph{IEEE Trans. Inf. Theory}, vol.~61, no.~10, pp.
  5553--5563, 2015.

\bibitem{soltani2018optical}
M.~Soltani and Z.~Rezki, ``Optical wiretap channel with input-dependent
  {G}aussian noise under peak-and average-intensity constraints,'' \emph{IEEE
  Trans. Inf. Theory}, vol.~64, no.~10, pp. 6878--6893, 2018.

\bibitem{dytso2018optimal}
A.~Dytso, M.~Egan, S.~M. Perlaza, H.~V. Poor, and S.~S. Shitz, ``Optimal inputs
  for some classes of degraded wiretap channels,'' in \emph{Proc. IEEE Inf.
  Theory Workshop}.\hskip 1em plus 0.5em minus 0.4em\relax IEEE, 2018, pp.
  1--5.

\bibitem{soltani2021degraded}
M.~Soltani and Z.~Rezki, ``The degraded discrete-time {P}oisson wiretap
  channel,'' \emph{arXiv preprint arXiv:2101.03650}, 2021.

\bibitem{nam2019secrecy}
S.-H. Nam and S.-H. Lee, ``Secrecy capacity of a {G}aussian wiretap channel
  with one-bit {ADC}s is always positive,'' in \emph{Proc. IEEE Inf. Theory
  Workshop}.\hskip 1em plus 0.5em minus 0.4em\relax IEEE, 2019, pp. 1--5.

\bibitem{smith1971information}
J.~G. Smith, ``The information capacity of amplitude-and variance-constrained
  scalar {G}aussian channels,'' \emph{Info. Control}, vol.~18, no.~3, pp.
  203--219, 1971.

\bibitem{DytsoAmplitute2020}
A.~Dytso, S.~Yagli, H.~V. Poor, and S.~{Shamai (Shitz)}, ``The capacity
  achieving distribution for the amplitude constrained additive {G}aussian
  channel: {A}n upper bound on the number of mass points,'' \emph{IEEE Trans.
  Inf. Theory}, vol.~66, no.~4, pp. 2006--2022, 2020.

\bibitem{karlin1957polya}
S.~Karlin, ``P{\'o}lya type distributions, ii,'' \emph{The Ann. Math. Stat.},
  vol.~28, no.~2, pp. 281--308, 1957.

\bibitem{Tijdeman1971number}
R.~Tijdeman, ``On the number of zeros of general exponential polynomials,'' in
  \emph{Indagationes Mathematicae (Proceedings)}, vol.~74.\hskip 1em plus 0.5em
  minus 0.4em\relax North-Holland, 1971, pp. 1--7.

\bibitem{esposito1968relation}
R.~Esposito, ``On a relation between detection and estimation in decision
  theory,'' \emph{Inf. Control}, vol.~12, no.~2, pp. 116--120, February 1968.

\bibitem{DytsoISIT2020extended}
\BIBentryALTinterwordspacing
A.~Dytso, H.~V. Poor, and S.~Shamai~(Shitz), ``A general derivative identity
  for the conditional mean estimator in {G}aussian noise and some
  applications,'' 2021. [Online]. Available:
  \url{https://arxiv.org/abs/2104.01883}
\BIBentrySTDinterwordspacing

\end{thebibliography}
\end{document}